\documentclass{llncs}

\usepackage[T1]{fontenc}
\usepackage[labelfont=bf]{caption}
\usepackage[labelfont=bf]{subcaption}

\let\doendproof\endproof
\renewcommand{\endproof}{\hfill$\qed$\doendproof}

\newif\ifllncsvar

\makeatletter \@ifclassloaded{llncs}{ \llncsvartrue

\newcommand{\rome}{$^*$} 
\newcommand{\kit}{$^{**}$}

\author{Giordano {Da Lozzo\rome} and Ignaz Rutter\kit 
\institute{
\rome~{Department of Engineering, Roma Tre University, Italy}\\  
\email{dalozzo@dia.uniroma3.it}\\ 
\kit~{Karlsruhe Institute of Technology (KIT), Germany}\\
\email{rutter@kit.edu}
}}

}{
  \author{Giordano {Da Lozzo}\thanks{Department of Engineering, Roma
      Tre University, Italy \texttt{dalozzo@dia.uniroma3.it} } \and
    Ignaz Rutter\thanks{Karlsruhe Institute of Technology (KIT),
      Germany \texttt{rutter@kit.edu}} }
}\makeatother

\usepackage[english]{babel}
\usepackage{tabularx,lipsum}
\usepackage{pdfsync}
\usepackage{amssymb}
\usepackage{xspace}
\usepackage{amsmath}
\usepackage{latexsym}
\usepackage{epsfig}
\usepackage{times}
\usepackage{enumerate}
\usepackage{caption}
\usepackage[plainpages=false]{hyperref}
\usepackage{multirow}
\usepackage[usenames,dvipsnames]{color}
\usepackage{hyphenat}
\usepackage{amsfonts}
\usepackage{dsfont}
\usepackage{paralist}
\usepackage{mfirstuc}
\usepackage[color=yellow,textsize=scriptsize]{todonotes}
\usepackage{marginnote}
\usepackage{cooltooltips}

\usepackage{tipa}
\usepackage[full]{textcomp}
\usepackage{ upgreek }
\usepackage{tabularx}
\usepackage{bold-extra}
\usepackage{framed}

\newcommand{\remove}[1]{}

\definecolor{blue}{rgb}{0.274,0.392,0.666}
\definecolor{red}{rgb}{0.627,0.117,0.156}
\newcommand{\NP}{$\mathcal{NP}$\xspace}

\newcommand{\NPC}{\mbox{\NP-complete}\xspace}

\newcommand{\NPCN}{\mbox{\NP-completeness}\xspace}

\newcommand{\NPHN}{\mbox{\NP-hardness}\xspace}

\newcommand{\sefe}{SEFE\xspace}
\newcommand{\Eg}[1]{\textcolor{OliveGreen}{$E^{#1}_3$}\xspace}

\newcommand{\Gr}{\textcolor{red}{$G_1$}\xspace}

\newcommand{\Gb}{\textcolor{blue}{$G_2$}\xspace}

\newcommand{\Gg}{\textcolor{OliveGreen}{$G_3$}\xspace}

\newcommand{\sunsefep}{{\scshape Sunflower SEFE}\xspace}

\newcommand{\sdrawing}{{$\omega$-streamed drawing}\xspace}
\newcommand{\SD}{{$\omega$-SD}\xspace}

\newcommand{\sBdrawing}{{$\omega$-streamed drawing with backbone}\xspace}
\newcommand{\SDB}{{$\omega$-SDB}\xspace}
\newcommand{\SDBp}[1]{{${#1}$-SDB}\xspace}

\newcommand{\spproblem}{{\sc Streamed Planarity}\xspace}
\newcommand{\spBproblem}{{\sc Streamed Planarity with Backbone}\xspace}
\newcommand{\spp}{{\sc SP}\xspace}
\newcommand{\spBp}{{\sc SPB}\xspace}

\newcommand{\instance}[1]{$\langle G(V_{#1}, S_{#1}), E_{#1}, \Psi_{#1} \rangle$\xspace}

\newcommand{\sefekinstance}[1]{${\langle G_i(V,E_i) \rangle }^{#1}_{i=1}$\xspace}

\newcommand{\sefeksolution}[1]{${\langle \mathcal{E}_i, A_i \rangle }^{#1}_{i=1}$\xspace}

\newcommand{\simplesefeksolution}[1]{${\langle \mathcal{E}_i\rangle }^{#1}_{i=1}$\xspace}

\newcommand{\algorithm}{{\sc ALGOCON}\xspace}
\newcommand{\algorithmbf}{{\sc \bf ALGOCON}\xspace}

\makeatletter
\@ifclassloaded{llncs}{
}{
\newtheorem{theorem}{Theorem}
\newtheorem{lemma}{Lemma}

}
\makeatother

\authorrunning{G. Da Lozzo and I. Rutter}
\titlerunning{Planarity of Streamed Graphs}
\title{Planarity of Streamed Graphs} 
\date{}

\begin{document}
\pagestyle{plain}
\maketitle

\begin{abstract}
  In this paper we introduce a notion of planarity for graphs that are
  presented in a streaming fashion.  A {\em streamed graph} is a
  stream of edges $e_1,e_2,\dots,e_m$ on a vertex set $V$.  A streamed
  graph is {\em $\omega$-stream planar} with respect to a positive
  integer window size $\omega$ if there exists a sequence of planar
  topological drawings $\Gamma_i$ of the graphs $G_i=(V,\{e_j \mid
  i\leq j < i+\omega\})$ such that the common graph
  $G^{i}_\cap=G_i\cap G_{i+1}$ is drawn the same in $\Gamma_i$ and in
  $\Gamma_{i+1}$, for $1\leq i < m-\omega$.  The {\sc Stream
    Planarity} Problem with window size $\omega$ asks whether a given
  streamed graph is $\omega$-stream planar.  We also consider a
  generalization, where there is an additional \emph{backbone graph}
  whose edges have to be present during each time step.  These
  problems are related to several well-studied planarity
  problems.
  
  We show that the {\sc Stream Planarity} Problem is \NPC even when
  the window size is a constant and that the variant with a backbone
  graph is \NPC for all $\omega \ge 2$.  On the positive side, we
  provide $O(n+\omega{}m)$-time algorithms for (i) the case $\omega =
  1$ and (ii) all values of $\omega$ provided the backbone graph
  consists of one $2$-connected component plus isolated vertices and
  no stream edge connects two isolated vertices. 
Our results improve on the Hanani-Tutte-style $O((nm)^3)$-time algorithm proposed by Schaefer~[GD'14] for $\omega=1$.
\end{abstract}

\section{Introduction}\label{se:introduction}

In this work we consider the following problem concerning the drawing of evolving networks.
We are given a stream of edges $e_1,e_2\dots,e_m$ with their endpoints
in a vertex set $V$ and an integer {\em window size} $\omega > 0$.
Intuitively, edges of the stream are assigned a fixed ``lifetime'' of
$\omega$ time intervals. Namely, for $1\leq i < |V|-\omega$, edge
$e_i$ will \emph{appear} at the $i$-th time instant and
\emph{disappear} at the $(i+\omega)$-th time instant.
We aim at finding a sequence of drawings $\Gamma_i$ of the graphs
$G_i=(V,\{e_j \mid i\leq j < i+\omega\})$, for $1\leq i < |V|-\omega$,
showing the vertex set and the subset of the edges of the stream that are ``alive'' at each time
instant $i$, with the following two properties: (i) each drawing
$\Gamma_i$ is planar and (ii) the drawing of the common graphs
$G^{i}_\cap=G_i\cap G_{i+1}$ is the same in $\Gamma_i$ and in
$\Gamma_{i+1}$.
We call such a sequence of drawings an {\em \sdrawing}
(\SD).

The introduced problem, which we call \spproblem (\spp, for short),
captures the practical need of displaying evolving relationships on the same
set of entities.
% 1
As large changes in consecutive drawings might negatively affect the
ability of the user to effectively cope with the evolution of the
dataset to maintain his/her mental map, in this model only one edge is allowed to enter the
visualization and only one edge is allowed to exit the visualization
at each time instant, visible edges are represented by the same curve during their lifetime, and each vertex is represented by the same distinct point. Thus, the amount of relational information displayed at any time stays constant.
However, the magnitude of information to be simultaneously presented to the user
may significantly depend on the specific application as well as on the nature of
the input data. Hence, an interactive visualization system would
benefit from the possibility of selecting different time windows. On
the other hand, it seems generally reasonable to consider time windows whose
size is fixed during the whole animation.

To widen the application scenarios, we consider the
possibility of specifying portions of a streamed graph that are alive
during the whole animation. These could be, e.g.,
context-related substructures of the input graph, like the backbone
network of the Internet (where edges not in the backbone disappear due
to faults or congestion and are later replaced by new ones), or sets
of edges directly specified by the user.  We call this variant of the
problem {\spBproblem} (\spBp, for short) and the sought sequence
  of drawings an \sBdrawing (\SDB).

\smallskip
\noindent\textbf{Related Work.}
The problem is similar to on-line planarity testing~\cite{dt-opt-96},
where one is presented a stream of edge insertions and deletions and
has to answer queries whether the current graph is planar.  Brandes
{\em et al.}~\cite{cbdanddgppsz-dtsm-12} study the closely related
problem of computing planar straight-line grid drawings of trees whose
edges have a fixed lifetime under the assumption that the edges are
presented one at a time and according to an Eulerian tour of the tree.
The main difference, besides using topological rather than
straight-line drawings, is that in our model the sequence of edges
determining the streamed graph is known in advance and no assumption
is made on the nature of the stream.

% Connection with SEFE
It is worth noting that the \spp Problem can be conveniently
interpreted as a variant of the much studied {\sc Simultaneous
  Embedding with Fixed Edges (SEFE)} Problem (see~\cite{bkr-sepg-12}
for a recent survey).  In short, an instance of SEFE
consists of a sequence of graphs $G_1,\dots,G_k$, sharing some
vertices and edges, and the task is to find a sequence of planar
drawings $\Gamma_i$ of $G_i$ such that $\Gamma_i$ and $\Gamma_j$
coincide on $G_i \cap G_j$.
It is not hard to see that deciding whether a streamed graph is
$\omega$-stream planar is equivalent to deciding whether the graphs
induced by the edges of the stream that are simultaneously present at
each time instant admit a SEFE.  Unfortunately, positive results on
SEFE mostly concentrate on the variant with $k=2$, whose complexity is
still open, and the problem is NP-hard for $k \ge 3$~\cite{gjpss-sgefe-06}.  However, while
the SEFE problem allows the edge sets of the input graphs to
significantly differ from each other, in our model only small changes
in the subsets of the edges of the stream displayed at consecutive
time instants are permitted.  In this sense, the problems we study can
be seen as an attempt to overcome the hardness of SEFE for $k \ge 3$
to enable visualization of graph sequences consisting of 
%more then two steps at least
several steps,
 when any two consecutive graphs exhibit a strong
similarity.

We note that the $\omega$-stream planarity of the stream
$e_1,\dots,e_m$ on vertex set $V$ and backbone edges $S$ is equivalent
to the existence of a drawing of the (multi)graph $p = (V,
\{e_1,\dots,e_m\} \cup S)$ such that (i) two edges cross only if
neither of them is in $S$ and (ii) if $e_i$ and $e_j$ cross, then
$|i-j| \ge \omega$.  As such the problem is easily seen to be a
special case of the {\sc Weak Realizability} Problem, which given a
graph $G=(V,E)$ and a symmetric relation $R \subseteq E \times E$ asks
whether there exists a topological drawing of $G$ such that no pair of
edges in $R$ crosses.  It follows that \spp and \spBp are contained in
$\mathcal NP$~\cite{sss-rsgnp-03}.  For $\omega=1$, the problem
amounts to finding a drawing of $un$, where a subset of the edges,
namely the edges of $S$, are not crossed.  This problem has recently
been studied under the name {\sc Partial
  Planarity}~\cite{abddgmpt-dnpgc-13,s-ppedgps-tr-13}.  Angelini et
al.~\cite{abddgmpt-dnpgc-13} mostly focus on straight-line drawings,
but they also note that the topological variant can be solved
efficiently if the non-crossing edges form a 2-connected graph.
Recently Schaefer~\cite{s-ppedgps-tr-13} gave an $O((nm)^3)$-time
testing algorithm for the general case of {\sc Partial Planarity} via
a Hanani-Tutte style approach.  He further suggests to view the
relation $R$ of an instance of {\sc Weak Realizability} as a conflict
graph on the edges of the input graph and to study the complexity
subject to structural constraints on this conflict graph.

\smallskip
\noindent\textbf{Our Contributions.}
In this work, we study the complexity of the \spp and \spBp Problems.
In particular, we show the following results.
\begin{compactenum}
\item \spBp is \NPC for all $\omega \geq 2$ when the backbone
  graph is a spanning tree.
\item There is a constant $\omega_0$ such that \spp with window size~$\omega_0$ is
\NPC.
\item We give an efficient algorithm with running time $O(n+\omega m)$
  for \spBp when the backbone graph consists of one 2-connected
  component plus, possibly, isolated vertices and no stream edge
  connects two isolated vertices.
\item We give an efficient algorithm for \spBp with running time $O(n+m)$ for
  $\omega=1$.
\end{compactenum}
It is worth pointing out that the second hardness result shows that
{\sc Weak Realizability} is \NPC even if the conflict graph
describing the non-crossing pairs of edges has bounded degree, i.e.,
every edge may not be crossed only by a constant number of other
edges.  In particular, this rules out the existence of FPT algorithms
with respect to the maximum degree of the conflict graph unless
$\mathcal P = \mathcal NP$.

For the positive results, note that the structural restrictions on the
variant for arbitrary values of $\omega$ are necessary to overcome the
two hardness results and are hence, in a sense, best possible.
Moreover, the algorithm for $\omega=1$ improves the previously best
algorithm for {\sc Partial Planarity} by
Schaefer~\cite{s-ppedgps-tr-13} (with running time $O((nm)^3)$-time)
to linear.  Again, since the problem is hard for all $\omega
\ge 2$, this result is tight.%  

%%%%%%%%%%%%%%%%%%%%%%%%%%%%%%%%%%%%%%%%%%%%%%%%%%%%%%%%%%%%%%%%%%%%%%%%%%%%%%%%%%%%%%%%%%%

\section{Preliminaries}\label{se:preliminaries}
% Planar graphs and embeddings
For standard terminology about graphs, drawings, and embeddings refer to~\cite{dett-gd-99}.

\remove{
A \emph{drawing} of a graph is a mapping of each vertex to a distinct
point of the plane and of each edge to a simple Jordan curve
connecting its endpoints. A drawing is \emph{planar} if the curves
representing its edges do not cross except, possibly, at common
endpoints. A graph is \emph{planar} if it admits a planar drawing. Two
drawings of the same graph are \emph{equivalent} if they determine the
same circular ordering of edges around each vertex. A \emph{planar
  embedding} is an equivalence class of planar drawings. A planar
drawing partitions the plane into topologically connected regions,
called \emph{faces}. %The unbounded face is the \emph{outer face}.
}

%\red{
	Given a $(k- 1)$-connected graph $G$ with $k\geq 1$, we denote by $k(G)$ the number of its maximal $k$-connected subgraphs. The maximal $2$-connected subgraphs are called {\em blocks}. Also, a $k$-connected component is {\em trivial} if it consists of a single vertex.
	Further, given a simply connected graph $G$, that is $1(G)=1$, the {\em block-cutvertex tree} $T$ of $G$ is the tree whose nodes are the cutvertices and the blocks of $G$, and whose edges connect nodes representing cutvertices with nodes representing the blocks they belong to. 	 

%Contracting an edge $(u,v)$ to a vertex $w$ in a graph $G$ is the operation of first removing edge $(u,v)$ from $G$, then adding $w$ to $G$ and making it adjacent to all the vertices $u$ and $v$ used to be adjacent to, except for $u$ and $v$, and finally removing multiple edges.

Contracting an edge $(u,v)$ in a graph $G$ is the operation of first removing $(u,v)$ from $G$, then identifying $u$ and $v$ to a new vertex $w$, and finally removing multi-edges.
%}

Let $G$ be a planar graph and let $\mathcal{E}$ be a planar
embedding of $G$. Further, let $H$ be a subgraph of $G$.
%, possibly a single edge. 
We denote by $\mathcal{E}{|_H}$ the embedding of $H$
determined by $\mathcal{E}$.

Let \sefekinstance{k} be $k$ planar graphs on the same set $V$ of
vertices. A {\em simultaneous embedding with fixed edges} ({\em SEFE})
of graphs $\langle G(V,E_i)\rangle^k_{i=1}$ consists of $k$ planar
embeddings $\langle \mathcal{E}_i\rangle^k_{i=1}$ such that $\mathcal{E}_i|_{G_{ij}}=\mathcal{E}_j|_{G_{ij}}$, with $G_{ij}= (V, E_i \cap E_j)$ for $i \ne j$.  The {\sc SEFE} Problem
corresponds to the problem of deciding whether the $k$ input graphs
admit a {\em SEFE}. Further, if all graphs share the same set of edges ({\em sunflower intersection}), that is, the graph $G_\cap = (V, E_i \cap E_j)$ is the same for
every $i$ and $j$, with $1 \leq i < j \leq k$, the problem is called
{\sc Sunflower SEFE} and graph $G_\cap$ is the {\em common graph}.

% Plane problem. Input description.
In the following, we denote a streamed graph by a triple \instance{} such that $G(V,S)$ is a planar graph, called {\em backbone graph}, $E
\subseteq V^2 \setminus S$ is the set of edges of a stream $e_1,e_2,\dots,e_m$, and
$\Psi: E \leftrightarrow \{1,\dots,m\}$ is a bijective function that
encodes the ordering of the edges of the stream.
% E_i

% G_\cup: Union Graph
Given an instance $I=$~\instance{}, we call graph $G_\cup=(V,
S \cup E)$ the {\em union graph} of $I$. Observe that, if $G_\cup$ has
$k$ connected components, then $I$ can be efficiently decomposed into
$k$ 
%(almost)\todo{actually we cannot just split the stream!}\ 
independent smaller instances, whose Streamed Planarity can be tested independently. Hence, in the following we will
only consider streamed graphs with connected union graph.
% Isolated vertices
Also, we denote by $\mathcal{Q}$ the set of isolated
vertices of $G$.

% NECESSITY and SUFFICIENCY
Note that, an obvious necessary condition for a streamed graph
\instance{} to admit an \SDB is the existence of a planar
combinatorial embedding $\mathcal{E}$ of the backbone graph $G$ such
that the endpoints of each edge of the stream lie on the boundary of
the same face of $\mathcal{E}$, as otherwise a crossing between an
edge of the stream and an edge of $G$ would occur.  However, since
each edge of the stream must be represented by the same
curve at each time, this condition is generally not
sufficient, unless $\omega=1$; see
Fig.~\ref{fig:ConsistentEdgeDrawing}.

\begin{figure}[tb]
  \centering
  \begin{subfigure}{.3\textwidth}
    \centering
    \includegraphics[page=1]{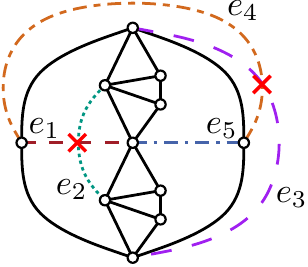}
    \caption{}
  \end{subfigure}
  \hspace{1cm}
  \begin{subfigure}{.3\textwidth}
    \centering
    \includegraphics[page=2]{img/NEG-PAR}
    \caption{}
  \end{subfigure}
  % \subfigure[]{\includegraphics[width=0.25\textwidth]{img/ConsistentEdgeDrawingB}}
  % \hspace{5mm}
  % \subfigure[]{\includegraphics[width=0.25\textwidth]{img/ConsistentEdgeDrawingA}}
\caption{Illustration of an instance \instance{} of \spBp with $\omega=2$, where $G$ is a $2$-connected graph, $E=\{e_i: 1\leq i \leq 5\}$, and $\Psi(e_i)=i$. Solid edges belong to $G$.  (a) and (b) show different embeddings of $G$ and assignments of the edges in $E$ to the faces of such embeddings. (a) determines a \SDBp{2} of \instance{}, while (b) does not.
%A crossing between $e_2$ and one of $e_1$ and $e_3$ must occur. (b-c) 
%\sdrawings of pair (b) $\langle G,\{e_1,e_2\} \rangle$ and (c) $\langle 
%G,\{e_2,e_3\} \rangle$.
}\label{fig:ConsistentEdgeDrawing}
\end{figure}

\section{Complexity}\label{se:npc}

In the following we study the computational complexity of testing
planarity of streamed graphs with and without a backbone graph.
First, we show that \spBp is \NPC, even when the backbone graph is a
spanning tree and $\omega = 2$.  This implies that \sunsefep is
\NPC for an arbitrary number of input graphs, even if every
graph contains at most $\xi = 2$ exclusive edges.
Second, we show that \spp is \NPC even for a constant window size
$\omega$.  This also has connections to the fundamental {\sc Weak Realizability}
Problem.
Namely,  Theorem~\ref{th:np-omega0} implies the \NPCN of {\sc Weak Realizability} even for instances
$\langle G(V,E), R \rangle$ such that the maximum number of
occurrences $\theta$ of each edge of $E$ in the pairs of edges
in $R$ is bounded by a constant, i.e., for each edge there
is only a constant number~$\theta$ of other edges it may not cross.

These results imply that, unless P=NP, no FPT algorithm with respect
to $\omega$, to $\xi$, or to $\theta$ exists for {\sc Streamed
  Planarity (with Backbone)}, \sefe, and {\sc Weak
  Realizability} Problems, respectively.

\begin{theorem}\label{th:nptree}
  \spBp is \NPC for $\omega \geq 2$, even when the backbone graph is a tree and the edges of
  the stream form a matching.
\end{theorem}

\begin{proof}
	The membership in \NP follows from~\cite{sss-rsgnp-03}.
	The \NPHN is proved by means of a polynomial-time reduction from
	problem \sunsefep, which has been proved \NPC for $k=3$ graphs, even
	when the common graph is a tree $T$ and the exclusive edges of each
	graph only connect leaves of the tree~\cite{adn-osnsp-14}. 
	
	Given an
	instance \sefekinstance{3} of \sunsefep, we construct a streamed
	graph \instance{} that admits an \SDB for $\omega=2$ if and only if
	\sefekinstance{3} is a positive instance of \sunsefep, as follows.
	% \remove{ To simplify the construction, we first reduce instance
	%   \ptcTRHEEpbeInstance{} of \ptcTRHEEbeshort to an equivalent
	%   instance \ptcTRHEEpbeInstance{*} of \ptcTRHEEbeshort in which the
	%   edges of \Er{*}$\cup$\Eb{*}$\cup$\Eg{*} form a matching.  Tree
	%   $T^*$ can be obtained from $T$ as follows. For each leaf $v$ of
	%   $T$ which is incident to $k$ edges in \Er{*}
	%   $\cup$\Eb{*}$\cup$\Eg{*}, we add to $T$ a set ${\cal L}(v)$ of $k$
	%   new leaves and connect them to $v$. Further, sets $E^*_i$ can be
	%   constructed as follows. For each edge $(u,v)\in E_i$, add to
	%   $E^*_i$ an edge between a vertex in ${\cal L}(u)$ and a vertex in
	%   ${\cal L}(v)$, in such a way that each leaf of $T^*$ is incident
	%   to at most one edge in \Er{*} $\cup$\Eb{*}$\cup$\Eg{*}.  It is
	%   easy to see that the two instances are equivalent.  }
	To simplify the construction, we first replace instance
	\sefekinstance{3} of \sunsefep with an equivalent instance in which
	the exclusive edges in $E_1 \cup E_2 \cup E_3$ form a matching, by
	applying the technique described in~\cite{adfpr-bicosefe-14}.  Then,
	we perform the reduction starting from such a new instance. Refer to
	Fig.~\ref{fig:3Page}.
	
	\begin{figure}[tb]
		% x = 1/(1+(17,3/13,2))
		\centering \centering
		\begin{subfigure}{.39\textwidth}
			\includegraphics[page=1, width=\textwidth]{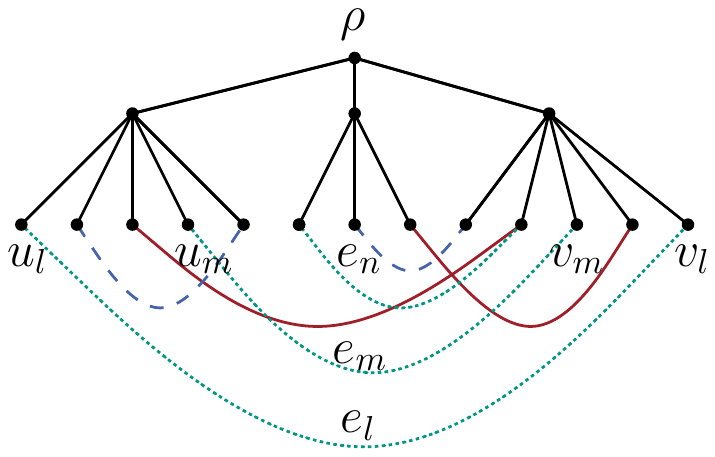}
			\caption{}
		\end{subfigure}
		\hspace{1cm}
		\begin{subfigure}{.51\textwidth}
			\includegraphics[page=2, width=\textwidth]{img/3pageTEMPLATE}
			\caption{}
		\end{subfigure}
		\caption{Illustration for the proof of
			Theorem~\ref{th:nptree}. (a) Instance \sefekinstance{3}. (b)
			Partial representation of instance \instance{} containing the
			edges of $G$ and the edges of the stream constructed starting
			from pairs of edges of \Eg{}. Edges of $T$ and $G$ are black,
			edges of \Gr{}, \Gb{}, and \Gg are solid red, dashed blue, and
			dotted green, respectively.  }\label{fig:3Page}
	\end{figure}

	% As we aim at producing an instance of \spBp such that the backbone
	% graph $G$ is connected, we show how to construct the vertex set
	% $V$ and the edge set $S$, by constructing graph ${\cal B}$ as
	% follows.
	First, set $G=T$. Then, for $i=1,2,3$ and for each edge $e = (u,v)
	\in E_i$, add to $G$ a star graph\footnote{A star graph is a tree
		with one internal node, called the {\em central vertex} of the
		star, and $k$ leaves.} $S(u_e)$ with leaves $u^1_e,\dots,u^q_e$
	and a star graph $S(v_e)$ with leaves $v^1_e,\dots,v^q_e$ with
	$q=|E_i|-1$, and identify the center of $S(u_e)$ with $u$ and the
	center of $S(v_e)$ with $v$, respectively. Also, consider the vertex
	$\rho$ of $G$ corresponding to any internal node of $T$, add to $G$
	vertices $s_{i}$, for $i=1,\dots,6$ (\emph{sentinel leaves}), and
	connect each of such vertices to $\rho$. Observe that, by
	construction, $G$ is a tree and $T\subset G$.  The sentinel edges
	will serve as endpoints of edges of the stream, called
	\emph{sentinel edges}, used to split the stream in three substreams
	in such a way that no edge of one substream is alive together with
	an edge of a different substream.
	
	Further, set $E$ can be constructed as follows. For $i=1,2,3$ and
	for each pair $\langle l, m\rangle$ of edges in $E_i$, add to $E$ an
	edge $lm=(u^a_l,v^a_l)$ between a leaf of $S(u_l)$ and a leaf of
	$S(v_l)$ and an edge $ml=(u^b_m,v^b_m)$ between a leaf of $S(u_m)$
	and a leaf of $S(v_m)$, respectively, for some $a,b \in
	{1,2,\dots,|E_i|-1}$, in such a way that no two edges in $E$ are
	incident to the same leaf of $G$.  Observe that, by construction,
	$E$ is a matching.  Also, add to $E$ edges $(s_1,s_2)$, $(s_3,s_4)$,
	and $(s_5,s_6)$ (\emph{sentinel edges}).
	
	Function $\Psi$ can be defined as follows. First, we construct an
	auxiliary ordering $\sigma=e_h,\dots,e_g$ of the edges in $E$, then
	we just set $\Psi(e)=\sigma(e)$, for any edge $e \in E$, where
	$\sigma(e)$ denotes the position of $e$ in $\sigma$.  To obtain
	$\sigma$, we consider sets $E_1$, $E_2$, and $E_3$ in this order and
	perform the following two steps. {\sc STEP 1:} for each pair
	$\langle l,m\rangle$ of edges in $E_i$, add to $\sigma$ edge $lm$
	and edge $ml$.  {\sc STEP 2:} add to $\sigma$ the sentinel edge
	$(v_{2(i-1)+1},u_{2(i-1)+2})$. Observe that, by construction, each
	common graph $G^i_\cap$ contains the edges of $G$ plus at most two
	edges $lm$ and $ml$ of the stream with $l,m \in E_i$, for some $i
	\in \{1,2,3\}$.
	
	Observe that, the reduction can be easily performed in polynomial
	time. 
	
	We now shot that \sefekinstance{3} admits a SEFE if and only if instance \instance{} admits an \SDB for $\omega=2$.
	
	Suppose that \sefekinstance{3} admits a SEFE
	\simplesefeksolution{3}. Let $\mathcal{H}$ be the embedding of the
	common graph $T$ in \simplesefeksolution{3}, that is,
	${\mathcal{H}}=\mathcal{E}_1|_T=\mathcal{E}_2|_T=\mathcal{E}_3|_T$. We
	construct a planar embedding $\mathcal{E}$ of $G$ by defining the
	rotation scheme of each non-leaf vertex of $G$, as follows.
	
	% by defining the ordering of
	% the children of each non-leaf vertex $v$ of $G$, as follows.  If
	% $v\in G \cap T^*$ and
	If $v$ is not a leaf of $T$, then the rotation scheme of $v$ in
	$\mathcal{E}$ is equal to the rotation scheme of $v$ in
	$\mathcal{H}$. If $v=u_l$ ($v=v_l$) is the unique neighbor of of any
	leaf vertex of $G$, then the rotation scheme of $u_l$ ($v_l$) can be
	chosen in such a way that the ordering of the leaves of $G$ that are
	adjacent to $u_l$ ($v_l$) is the reverse of the ordering of the
	leaves of $G$ that are adjacent to $v_l$ ($u_l$), where the the
	leaves of $G$ that are adjacent to $u_l$ ($v_l$) and to $v_l$
	($u_l$) are identified by the corresponding apex.
	We claim that the constructed embedding $\mathcal{E}$ of $G$ yields
	an \SDB of \instance{} for $\omega=2$.  Let $\mathcal{O}$ be the
	circular ordering of the leaves of $T$ determined by an Eulerian
	tour of $T$ in $\mathcal{H}$. Also, let $\mathcal{O}'$ be the
	circular ordering of the leaves of $G$ determined by an Eulerian
	tour of $G$ in $\mathcal{E}$.  Suppose that there exist two edges
	$xy$ and $yx$ with $|\Psi(xy)-\Psi(yx)|<\omega=2$ such that the
	endpoints $u^i_x$ and $v^i_x$ of edge $xy$ and the endpoints $u^j_y$
	and $v^j_y$ of edge $yx$ alternate in $\mathcal{O}'$. This implies
	that the unique neighbors $u_x$ of $u^i_x$, $v_x$ of $v^i_x$, $u_y$
	of $u^j_y$, and $v_y$ of $v^j_y$ in $T$ alternate in
	$\mathcal{O}$. This, in turn, implies a crossing between the two
	edges $x$ and $y$ of some set $E_i$. Hence, contradicting the fact
	that \simplesefeksolution{3} is a SEFE.
	
	Suppose that \instance{} admits an \SDB for $\omega=2$. Let
	$\mathcal{E}$ be the planar embedding of $G$ in any \SDB of
	\instance{}. Let $\mathcal{O}$ be the ordering of the leaves of $G$
	in an Eulerian tour of $G$ in $\mathcal{E}$. Also, let
	$\mathcal{O}'$ of the ordering of the leaves of $T$ in an Eulerian
	tour of $T$ in the embedding $H=\mathcal{E}|_{T}$. We claim that $H$
	yields a SEFE of \sefekinstance{3}.  Suppose that there exist two
	edges $x=(u_x,v_x)$ and $y=(u_y,v_y)$ of some set $E_i$ whose
	endpoints alternate in $\mathcal{O}'$. Consider the two edges
	$xy=(u^p_x,v^p_x)$ and $yx=(u^q_y,v^q_y)$ of $E$, with $1\leq p \leq
	|E^*_i|-1$ and $1 \leq q \leq |E^*_i|-1$.  Since the sets of leaves
	of $S(u_x)$, $S(v_x)$, $S(u_y)$, and $S(v_y)$ appear in
	$\mathcal{O}$ in the same order as the vertices $u_x$, $v_x$, $u_y$,
	and $v_y$ appear in $\mathcal{O}'$, the endpoints of $xy$ and $yx$
	alternate in $\mathcal{O}'$. Further, by construction, it holds that
	either $\Psi(xy)=\Psi(yx)+1$ or $\Psi(yx)=\Psi(xy)+1$, that is,
	either edge $xy$ immediately precedes edge $yx$ in the stream or
	edge $yx$ immediately precedes edge $xy$ in the stream. The above
	facts then imply a crossing between edge $xy$ and $yx$ of the
	stream. Hence, contradicting the hypothesis that \instance{} admits
	an \SDB for $\omega=2$. 
	
	The above discussion proves the statement for $\omega = 2$. To extend the theorem to any value of $\omega \geq 2$ it suffices to augment \instance{} with additional sentinel leaves and sentinel edges. This concludes the proof of the theorem.
\end{proof}

\begin{theorem}\label{th:np-omega0}
\label{th:np-omega0}
 There is a constant $\omega_0$ such that deciding whether a given
  streamed graph is $\omega_0$-stream planar is \NPC.
\end{theorem}

\begin{proof} The membership in \NP follows from~\cite{sss-rsgnp-03}.  In the
	following we describe a reduction that, given a 3-SAT formula
	$\varphi$, produces a streamed graph that is $\omega_0$-stream
	planar if and only if $\varphi$ is satisfiable.
	
	To make things simple, we do not describe the stream, but rather
	important keyframes.  Our construction has the property that edges
	have a FIFO behavior, i.e., if edge $e$ appears before edge $f$,
	then also $e$ disappears before $f$.  This, together with the fact
	that in each key frame only $O(1)$ edges are visible ensures that
	the construction can indeed be encoded as a stream with window
	size $O(1)$.  The value $\omega_0$ we use is simply the maximum
	number of visible edges in any of the key frames.  We do not take
	steps to further minimize $\omega_0$, but even without this, the
	value produced by the reduction is certainly less than 120, as we
	estimate at the end of the proof.  Sometimes, we wish to wait
	until a certain set of edges has disappeared.  In this case we
	insert sufficiently many isolated edges into the stream, which
	does not change the $\omega_0$-planarity of the stream.
	
	We now sketch the construction.  It consists of two main pieces.
	The first is a cage providing two faces called \emph{cells}, one
	for vertices representing satisfied literals and one for vertices
	representing unsatisfied literals.  We then present a clause
	stream for each clause of $\varphi$.  It contains one literal
	vertex for each literal occurring in the clause and it ensures
	that these literal vertices are distributed to the two cells of
	the cage such that at least one goes in the cell for satisfied
	literals.  Throughout we ensure that none of the previously
	distributed vertices leaves the respective cell.
	
	Second, we present a sequence of edges that is $\omega_0$-stream
	planar if and only if the previously chosen distribution of the
	literal vertices forms a truth assignment.  This is the case if
	and only if any two vertices representing the same literal are in
	the same cell and any two vertices representing complementary
	literals of one variable are in distinct cells.
	
	It is clear that, if the constructions work as described, then the
	resulting streamed graph is $\omega_0$-stream planar if and only
	if $\varphi$ is satisfiable.  The first part of the stream ensure
	that from each clause one of the literals must be assigned to the
	cell containing satisfied literals (i.e. the literal receives the
	value true).  The second part ensures that these choices are
	consistent over all literals, i.e., these choices actually
	correspond to a truth assignment of the variables.
	
	\begin{figure}[tb] \centering
		\includegraphics{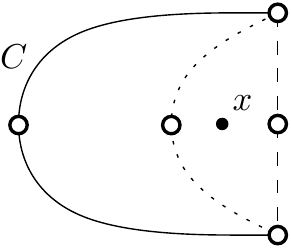}
		\caption{Cycle $C$ (solid and dashed edges) contains vertex $x$
			in its interior.  The dashed edges leave the sliding window
			soon.  Presenting a new path (dotted) parallel to the old path
			does not ensure that $x$ ends up in the interior of the
			resulting cycle $C'$ (solid and dotted edges).}
		\label{fig:edge-traversal}
	\end{figure}
	
	Our first step will be the construction of the cage containing the
	two cells.  Since the cage needs to persist throughout the whole
	sequence, it must be constructed in such a way that it can be
	``kept alive'' over time by presenting new edges.  Note that it
	does not suffice to repeatedly present edges that are parallel to
	existing ones, as they may be embedded differently, and hence over
	time allow isolated vertices to move through obstacles; see
	Fig.~\ref{fig:edge-traversal}.  We first present a construction
	that behaves like an edge that can be ``renewed'' without changing
	its drawing too much.  We call it \emph{persistent edge}.
	
	\begin{figure}[tb] \centering
		\begin{subfigure}{.32\textwidth} \centering
			\includegraphics[page=1]{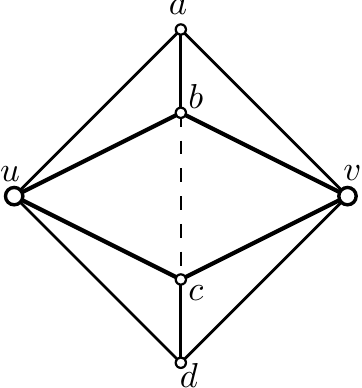}
			\caption{}
			\label{fig:persistent-edge-1}
		\end{subfigure} \hfill
		\begin{subfigure}{.32\textwidth} \centering
			\includegraphics[page=3]{fig/persistent-edge}
			\caption{}
			\label{fig:persistent-edge-2}
		\end{subfigure} \hfill
		\begin{subfigure}{.32\textwidth} \centering
			\includegraphics[page=4]{fig/persistent-edge}
			\caption{}
			\label{fig:persistent-edge-3}
		\end{subfigure}
		\caption{A persistent edge.  The thickness of the edges
			indicates how long the edge stays in the sliding window.  The
			thinner the edge the earlier it leaves the window.  (a) The
			initial configuration; the dashed edge $bc$ dissolves first.
			It is used only once to initially enforce a unique planar
			embedding.  (b) New vertices $b'$ and $c'$ with neighbors
			$u,b,v$ and $u,c,v$, respectively, are introduced.  Starting
			from the embedding in (a) the embedding is uniquely defined.
			(c) After the edges incident to $a$ and $d$ disappear, the
			drawing has again the same structure as in (a).  Repeating
			this cycle hence preserves the edge.  Since edges are embedded
			only in the interior of the gadget vertices that are embedded
			outside the persistent edge cannot traverse it.}
		\label{fig:persistent-edge}
	\end{figure}
	
	Let $u$ and $v$ be two vertices.  A persistent edge between $u$
	and $v$ consists of the four vertices $a,b,c,d$, each lying on a
	path of length~2 from $u$ to $v$.  Additionally, $a$ is connected
	to $b$ and $b$ is connected to $c$.  Initially, we also have
	insert the edge $b,c$ to enforce a unique planar embedding.
	However, once it leaves the sliding window it does not get
	replaced.  Figure~\ref{fig:persistent-edge-1} shows a persistent
	edge where the thickness of the edge visualizes the time until an
	edge leaves the sliding window.  The thicker the edge the longer
	it stays.  Once the edge $bc$ has been removed, but before any of
	the other edges disappear, we present in the stream the edges
	$ub'$, $vb'$ and $bb'$ as well as $uc'$, $vc'$ and $cc'$, where
	$b'$ and $c'$ are new vertices; see
	Fig.~\ref{fig:persistent-edge-2}.  Note that there is a unique way
	to embed them into the given drawing.  After the edges $ua$, $av$
	leave the sliding window, $b$ takes over the role of $a$ and $b'$
	takes over the role of $b$.  Similarly, after the edges $ud$ and
	$dv$ leave the sliding window, $c$ takes over the role of $d$ and
	$c'$ takes over the role of $c$; see
	Fig.~\ref{fig:persistent-edge-3}.  By presenting six new edges in
	regular intervals, the persistent edge essentially keeps its
	structure.  In particular, we know at any point in time which
	vertices are incident to the inner and outer face.  For simplicity
	we will not describe in detail when to perform this book keeping.
	Rather, we just assume that the sliding window is sufficiently
	large to allow regular book keeping.  For example, before each of
	the steps described later, we might first update all persistent
	edges, then present the gadget performing one of the steps, then
	update the persistent edges gain, and finally wait for the gadget
	edges to be removed from the sliding window again.
	
	\begin{figure}[tb] \centering
		\begin{subfigure}[b]{.3\textwidth} \centering
			\includegraphics[page=1]{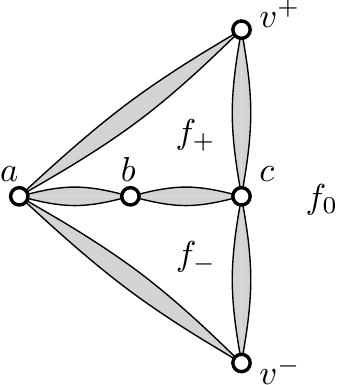}
			\caption{}
			\label{fig:cage}
		\end{subfigure}
		\begin{subfigure}[b]{.3\textwidth} \centering
			\includegraphics[page=3]{fig/cage}
			\caption{}
			\label{fig:cage-check-same}
		\end{subfigure}
		\begin{subfigure}[b]{.3\textwidth} \centering
			\includegraphics[page=2]{fig/cage}
			\caption{}
			\label{fig:cage-check-different}
		\end{subfigure}
		
		\caption{(a) The cage, the thick gray edges are persistent edges
			and are refreshed at regular intervals.  After presenting all
			clause sequences, the faces $f^+$ and $f^-$ will contain the
			literal vertices corresponding to satisfied and unsatisfied
			literal vertices, respectively. (b) Edges used to check
			whether two literal vertices $x_i$ and $x_j$ are in the same
			face. (c) Edges used to check whether literal vertices
			$\overline{x_i}$ and $x_j$ are in distinct faces.}
		\label{fig:cage-and-check}
	\end{figure}
	
	Next, we describe the cage.  Conceptually, it consists of two
	cycles of length~4, on vertices $a,b,c,v^+$ and $a,b,c,v^-$,
	respectively.  However, the edges are actually persistent edges;
	see Fig.~\ref{fig:cage}.  The interior faces $f^+$ and $f^-$ of
	the two cycles are the positive and negative literal faces,
	respectively.  Note that at any point in time only a constant
	number of edges are necessary for the cage.
	
	Before we describe the clause gadget, which is the most involved
	part of the construction, we briefly show how to perform the test
	for the end of sequence.  Namely, assume that we have a set $V'
	\subseteq V$ of literal vertices, and each of them is contained in
	one of the two literal faces. More formally, for each clause $c_i
	\in \varphi$ and for each Boolean variable $x$, set $V'$ contains
	a literal vertex $x_i$, if $x \in c_i$, or a literal vertex
	$\overline{x_i}$, if $\overline{x} \in c_i$. To check whether two
	literal vertices $x_i$ and $x_j$ corresponding to a variable $x$
	are in the same face, it suffices to present an edge between them
	in the stream, then wait until that edge leaves the sliding
	window, and continue with the next pair; see
	Fig~\ref{fig:cage-check-same}.  Of course, in the meantime we may
	have to refresh the persistent edges.  Similarly, if we wish to
	check that literal vertices $\overline{x_i}$ and $x_j$ are in
	distinct faces, we make use of the fact that the two cycles
	forming the cage share two edges, and hence three vertices $a,b$
	and $c$.  We present in the stream the complete bipartite graph on
	the vertices $\{\overline{x_i},x_j\}$ and $\{a,b,c\}$.  Clearly,
	this can be drawn in a planar way if and only if $\overline{x_i}$
	and $x_j$ are in distinct faces; see
	Fig.~\ref{fig:cage-check-different}.  Again, it may be necessary
	to wait until these edges leave the sliding window before the next
	test can be performed.

	\begin{figure}[tb] \centering
		\begin{subfigure}[b]{.3\textwidth} \centering
			\includegraphics[page=2, scale=.85]{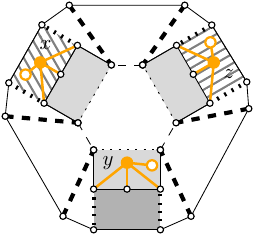}
			\caption{}
			\label{fig:clause-1}
		\end{subfigure}\hfill
		\begin{subfigure}[b]{.3\textwidth} \centering
			\includegraphics[page=4, scale=.85]{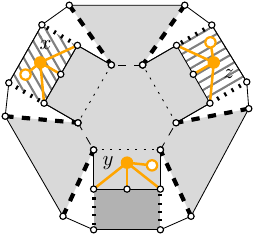}
			\caption{}
			\label{fig:clause-2}
		\end{subfigure}\hfill
		\begin{subfigure}[b]{.3\textwidth} \centering
			\includegraphics[page=6, scale=.85]{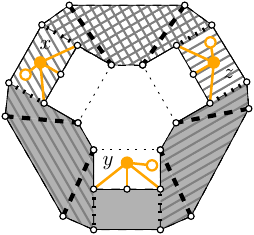}
			\caption{}
			\label{fig:clause-3}
		\end{subfigure}
		
		\begin{subfigure}[b]{.6\textwidth} \centering
			\includegraphics[page=7, scale=.85]{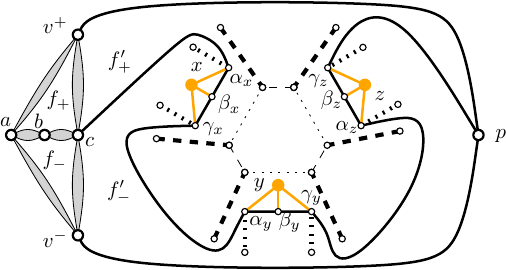}
			\caption{}
			\label{fig:clause-4}
		\end{subfigure}
		\begin{subfigure}[b]{.3\textwidth} \centering
			\includegraphics[page=8, scale=.85]{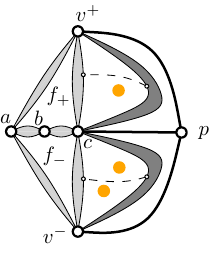}
			\caption{}
			\label{fig:clause-5}
		\end{subfigure}

		\caption{Illustration of the clause sequence. (a) Initial
			embedding of the clause.  (b), (c) faces indicator vertices
			can reach if they are embedded in the face close to the center
			and close to the boundary, respectively. (d), separating the
			vertices corresponding to satisfied and unsatisfied literals
			into two distinct faces. (e) Integrating the now separated
			literal vertices into the corresponding faces of the cage.}
		\label{fig:clause}
	\end{figure}
	
	Finally, we describe our clause gadget; see Fig.~\ref{fig:clause}
	for an illustration.  First, we present the clause gadget as it is
	shown in Fig.~\ref{fig:clause-1}.  The literal vertices are large
	and solid, their corresponding indicator vertices are represented
	by large empty disks.  The edges are ordered in the stream such
	that the three edges connecting a literal vertex to its indicator
	are presented first, i.e., they also leave the sliding window
	first.  The remaining three edges incident to the literals are
	drawn last so that they remain present longest.  Observe that the
	embedding of the clause without the literal and indicator vertices
	is unique; we call this part of the clause the \emph{frame}.  Each
	literal vertex may choose among two possible faces of the frame
	where it can be embedded.  Either close to the center or close to
	the boundary.  The faces in the center are shaded light gray, the
	faces on the boundary are shaded or tiled in a darker gray in
	Fig.~\ref{fig:clause-1}.
	
	We now first wait until the edges between literal vertices and
	their indicators leave the sliding window.  Now the following
	things happen.  First, the thin dotted and dashed edges leave the
	sliding window.  Immediately afterwards, we present in the stream
	paths of length~2 that replace these edges, so the frame
	essentially remains as it is shown.  However, after this step, the
	indicator vertex of any literal that was embedded in the face
	close to the center may be in any of the faces shaded in light
	gray in Fig.~\ref{fig:clause-2}.  Now, first the thick dotted
	edges leave the sliding window and are immediately replaced by
	parallel paths.  Afterwards, the thick dashed edges leave the
	sliding window and are immediately replaced by parallel paths.
	Again, the frame remains essentially present.  This allows the
	indicator vertices of literals that were embedded on the outer
	face to traverse into the faces indicated in
	Fig.~\ref{fig:clause-3}.  Note that, if all literal vertices were
	embedded in the face close the boundary, then there is no face of
	the frame that can simultaneously contain them at this point.  If
	however, at least one of them was embedded in the face close to
	the center, then there is at least one face of the frame that can
	contain all the vertices simultaneously.  We now include in the
	stream a triangle on the three indicator vertices.  This triangle
	can be drawn without crossing edges of the frame if and only if
	the three vertices can meet in one face, which is the case if and
	only if at least one indicator vertex, and hence also its
	corresponding literal vertex, was embedded close to the center.
	Now we wait until the edges of the clause, except for those
	incident to the literal vertices and the paths that were renewed
	have vanished; see Fig.~\ref{fig:clause-4}.
	
	Let now $p$ be a new vertex, and denote the neighbors of the
	literal vertex $x$ by $\alpha_x,\beta_x$ and $\gamma_x$, and
	similarly for $y$ and $z$.  We now connect $v$ to the cage by
	present the edges $v^-v$ and $v^+v$ as well as edges forming a
	path from $c$ to $p$ that, starting from $p$, first visits
	$\alpha_x,\beta_x,\gamma_x$, then $\alpha_y,\beta_y,\gamma_y$, and
	finally $\alpha_z,\beta_z,\gamma_z$.  Observe that the fact that
	$p$ has disjoint paths to $v^-,v+$ and $v$ containing the
	$\alpha_h,\beta_h$ and $\gamma_h$, with $h \in \{x,y,z\}$, ensures
	that, what remains of the clause gadget must be (and hence must
	have been all the time) embedded in the outer face of the cage.
	We assume without loss of generality that the path containing the
	$\alpha_h,\beta_h$ and $\gamma_h$, with $h \in \{x,y,z\}$, is not
	incident to the outer face.  Again, we consider the edges incident
	to the literal vertices not as part of the construction.  Then the
	path is incident to precisely two faces, which are adjacent to the
	literal faces of the cage.  Denote the one incident to $f_+$ by
	$f_+'$ and the one incident to $f_-$ by $f_-'$; see
	Fig.~\ref{fig:clause-4}.  Due to the traversal, we have that a
	literal vertex $v$ is contained in $f_+'$ if and only if it was
	embedded in the face close to the center in the clause, which
	means that the corresponding literal was satisfied.  Otherwise, it
	is embedded in $f_-'$.  It now remains to enclose the literal
	vertices into the corresponding face of the cage without letting
	escape any of the literal vertices already embedded there.
	
	First, we wait until all edges incident to the literal vertices
	have left the sliding window, i.e., they become isolated.  Then,
	we present two new persistent edges parallel to the existing
	persistent edges $v^+c$ and $v^-c$, respectively; see
	Fig.~\ref{fig:clause-5}, where the new persistent edges are shaded
	dark gray.  To ensure that the embedded is indeed as shown in
	Fig.~\ref{fig:clause-5}, we one boundary vertex of each new
	persistent edge to a vertex on the outer boundary of the
	persistent edge it is parallel to (dashed lines in
	Fig.~\ref{fig:clause-5}).  The new parallel edges replace the old
	persistent edges of the cage, and we wait until they have
	dissolved.  Clearly, no vertex from an internal face of the cage
	can escape as the new persistent edges are embedded in the outer
	face of the cage.  To ensure that the literal vertices must indeed
	be embedded in the literal faces, we present the edges $bx$, $by$
	and $bz$.  Finally, we wait until these edges vanish again.  Then
	we are ready for the next clause sequence or for the final
	checking sequence.
	
	The above description produces for a given 3SAT formula $\varphi$
	produces, for a sufficiently large (but constant!) $\omega_0$ a
	stream $S_\varphi$ one some vertex set $V_\varphi$ such that
	$\varphi$ is satisfiable if and only if $S_\varphi$ is
	$\omega_0$-stream planar.  In the first part of the stream, in any
	sequence of corresponding planar embedding, the literals of each
	clause, represented by vertices, are transferred to two interior
	faces of the cage such that for each clause at least one literal
	vertex is transferred to the face representing satisfied literals.
	This models the fact that each clause must contain at least one
	satisfied literal.  In the second part, a sequence of edges is
	presented that is $\omega_0$-planar if and only if the previously
	produced distribution of literals to the positive and negative
	faces of the cage corresponds to a truth assignment of the
	underlying variables.  The construction can clearly be performed
	in polynomial time.
	
	We now briefly estimate the window size~$\omega_0$. The largest
	number of edges that are simultaneously important in our
	construction occurs when presenting a clause gadget.  A clause
	gadget has 48 edges, and it is simultaneously visible with four
	persistent edges, each of which may use up to 16 edges immediately
	after they have refreshed.  Hence a window size of $\omega_0 =
	112$ suffices for the construction.
\end{proof}

%%%%%%%%%%%%%%%%%%%%%%%%%%%%%%%%%%%%%%%%%%%%%%%%%%%%%%%%%%%%%%%%%%%%%%%%%%%%%%%%%%%%%%%%%

\section{Algorithms for $\omega$-Stream Drawings with
  Backbone}\label{se:polynomiality}
In this section, we describe a polynomial-time decision algorithm for
the case that the backbone graph consists of a $2$-connected component
plus, possibly, isolated vertices with no edge of the stream
connecting two isolated vertices.
We call instances satisfying these properties {\em star instances}, as
the isolated vertices are the centers of edge disjoint star subgraphs
of the union graph (see Section~\ref{se:star}). 
Observe that, the requirement of the absence of edges of the stream
between the isolated vertices of a star instance seems to be quite a
natural restriction. In fact, as proved in Theorem~\ref{th:np-omega0},
dropping this restriction makes the \spproblem Problem computationally tough.
This algorithm will also serve as a subprocedure to solve the {\spBp Problem} for
$\omega=1$ with no restrictions on the backbone graph (see
Section~\ref{se:unit}).  

\subsection{Star Instances}\label{se:star}
In this section we describe an efficient algorithm to test the
existence of an \SDB for star instances (see Fig.~\ref{fig:star}(a)).
\begin{figure}[tb]
  \centering
  \begin{subfigure}{0.48\textwidth}
    \centering
    \includegraphics[height=0.5\textwidth]{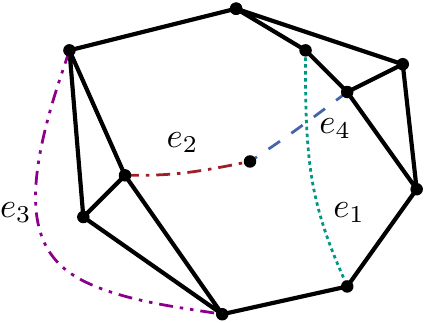}
  \end{subfigure}
  \begin{subfigure}{0.48\textwidth}
    \centering
    \includegraphics[height=0.5\textwidth]{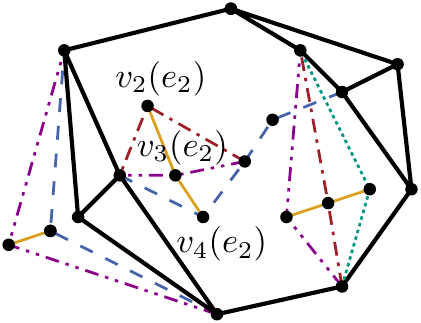}
  \end{subfigure}
  \caption{(a) A star instance with stream edges $E=\{e_i: 1\leq i \leq 4\}$, $\Psi(e_i)=i$, and $\omega=3$. (b) A SEFE of the instance of \sunsefep obtained as described in Lemma~\ref{le:starreduction} where $G_\cup$ is drawn with thick solid black edges, exclusive edges of
    $G_i$ are drawn with the same style as edge $e_i$ and exclusive
    edges of $G_{m+1}=G_{5}$ are drawn as yellow solid curves.
    Vertices in $D(e_2)=\{v_2(e_2),v_3(e_2),v_4(e_2)\}$ are also
    shown.}\label{fig:star}
\end{figure}
The problem is equivalent to finding an embedding $\mathcal{E}$ of the
unique non-trivial $2$-connected component  $\beta$ of $G$ and an assignment of the
edges of the stream and of the isolated vertices of $G$ to the faces
of~$\mathcal{E}$ that yield a \SDB.

\begin{lemma}\label{le:starreduction}
  Let \instance{} be a star instance of \spBp and let $\omega$ be a positive integer window size. There exists an
  equivalent instance \sefekinstance{m+1} of {\sc Sunflower SEFE} such
  that % (i) all graphs $G_i$ share the same common graph $G_\cap$,
       % that is, $(V,E_i\cap E_j)=G_i\cap G_j=G_\cap$, with $1\leq
       % i<j\leq m+1$ and (ii)
  the common graph $G_\cap$ consists of disjoint $2$-connected
  components. Further, instance \sefekinstance{m+1} can be constructed
  in $O(n+\omega{}m)$ time.
\end{lemma}

\begin{proof}
	Given a star instance \instance{} of \spBp we construct an instance
	\sefekinstance{m+1} of  \sunsefep that admits a SEFE if and only if
	\instance{} admits an \SDB, as follows.  Refer
	to~Figs~\ref{fig:star}(a) and ~\ref{fig:star}(b) for an example of
	the
	construction. % of graphs $G_\cap$ and $G_i$, for $i=1,\dots,m+1$.
	
	Initialize graph $G_\cap$ to the backbone graph $G$. Also, for every
	edge $e \in E$, add to $G_\cap$ a set of vertices $D(e) = \{v_i(e)
	\mid \Psi(e)\leq i < \min(\Psi(e)+\omega,m+1) \}$. Observe that,
	since \instance{} is a star instance, graph $G_\cap$ contains a
	single non-trivial $2$-connected component $\beta$, plus a set of
	trivial $2$-connected components consisting of the isolated vertices
	in $\mathcal{Q} \cup \bigcup_{e \in E} D(e)$.
	
	For $i=1,\dots,m$, graph $G_i$ contains all the edges and the
	vertices of $G_\cap$ plus a set of edges defines as follows. For
	each edge $e=(u,v) \in E$ such that $0 \leq i-\Psi(e)< \omega$, add
	to $E(G_i)$ edges $(u,v_i(e))$ and $(v_i(e),v)$. From a high-level
	view, graphs $G_i$, with $i=1,\dots,m$, are defined in such a way to
	enforce the same constraints on the possible embeddings of the
	common graph as the constraints enforced by the edges of the stream
	on the possible embeddings of the backbone graph.
	
	Finally, graph $G_{m+1}$ contains all the edges and the vertices of
	$G_\cap$ plus a set of edges defined as follows. For each edge $e
	\in E$, add to $E_{m+1}$ edges
	% $(v_k(e),v_{k+1}(e))$, with $\Psi(e) \leq k <\Psi(e)+\omega-1$
	$(v_{\Psi(e)}(e),v_{k}(e))$, with $\Psi(e) < k < min(\Psi(e)+\omega,
	m+1)$.
	% In the following, we denote by $G^\diamond$ graph $G_{m+1}$.
	Observe that, in any planar drawing $\Gamma_{m+1}$ of $G_{m+1}$,
	vertices $v_k(e)$ lie inside the same face of $\Gamma_{m+1}$, for
	any edge $e \in E$. The aim of graph $G_{m+1}$ is to combine the
	constrains imposed on the embedding of the backbone graph by each
	graph $G_i$, with $i=1,\dots,m$, in such a way that, for each edge
	$e \in E$, the edges of set $D(e)$ are embedded in the same face of
	the backbone graph.
	
	% Nomenclature
	Hereinafter, given a positive instance \sefekinstance{m+1} of SEFE
	with the above properties, we denote the corresponding SEFE $\langle
	\Gamma_i \rangle^{m+1}_{i=1}$ by \sefeksolution{m+1}, where
	$\mathcal{E}_i$ represents the embedding of $\beta$ in $\Gamma_i$
	and $A_{\mathcal{E}_i}$ represents the assignment of the isolated
	vertices and of the exclusive edges of graph $G_i$ in $\Gamma_i$ to the faces of
	$\mathcal{E}_i$, for $i=1,\dots,m+1$.
	Similarly, given a positive star instance \instance{} of \spBp we denote the corresponding \SDB $\Gamma$ by $\langle \mathcal{E}, A_{\mathcal{E}}\rangle$, where $\mathcal{E}$ represents the embedding of the unique non-trivial $2$-connected component $\beta$ of $G$ in $\Gamma$ and $A_\mathcal{E}$ represents the assignment of the isolated vertices of $G$ and of the edges of the stream to the faces
	of~$\mathcal{E}$ in $\Gamma$. More formally, $A_{\mathcal{E}} \colon E \cup \mathcal{Q} \rightarrow F(\mathcal{E})$, where $F(\mathcal{E})$ denotes the set of facial cycles of $\mathcal{E}$.

	%% First Direction
	Suppose that \sefekinstance{m+1} is a positive instance of SEFE,
	that is, \sefekinstance{m+1} admits a SEFE \sefeksolution{m+1}.  We
	show how to construct a solution $\langle \mathcal{E}, A_{\mathcal{E}}\rangle$ of \instance{}.

	Since \sefeksolution{m+1} is a SEFE and $\beta \in G_\cap$, we have
	that $\mathcal{E}_i= \mathcal{E}_j$, with $1\leq i<j \leq m+1$. We
	set the embedding $\mathcal{E}$ of $\beta$ to $\mathcal{E}_1$.

	Further, for every edge $e \in E$, we set $A_\mathcal{E}(e)$ to the
	face of $\mathcal{E}_1$ vertex $v_{\Psi(e)}(e)$ is placed inside in
	$\Gamma_1$, that is,
	$A_\mathcal{E}(e)=A_{\mathcal{E}_1}(v_{\Psi(e)}(e))$. Similarly, for
	every isolated vertex $v \in \mathcal{Q}$, we set $A_\mathcal{E}(v)$
	to the face of $\mathcal{E}_1$ vertex $v$ is placed inside in
	$\Gamma_1$, that is, $A_\mathcal{E}(v)=A_{\mathcal{E}_1}(v)$.
	
	We need to prove that $\mathcal{E}$ is a planar embedding of $\beta$
	and that no crossing occurs neither between an edge in $E$ and an
	edge in $\beta$ nor between two edges $e_i \in E$ and $e_j \in E$,
	with $i<j$ and $\Psi(e_j)-\Psi(e_i)<\omega$. Observe that, since
	\sefeksolution{m+1} is a SEFE, the embedding $\mathcal{E}_i$ of
	$\beta$ in $\Gamma_i$ is planar. As $\mathcal{E}$ coincides with
	$\mathcal{E}_1$, it follows that $\mathcal{E}$ is also
	planar. Assume that there exists a crossing between an edge $e \in
	E$ and an edge of $\beta$. This implies that there exists in
	$\Gamma_{\Psi(e)}$ a path $p^*=(u,v_{\Psi(e)}(e),v)$ connecting two
	vertices of $u$ and $v$ of $\beta$ that are incident to different
	faces of $\mathcal{E}_{\Psi(e)}$. Further, assume that there exists
	a crossing between an edge $e_i \in E$ and an edge $e_j \in E$ with
	$\Psi(e_i)<\Psi(e_j)$ such that $\Psi(e_j) - \Psi(e_i)< \omega$
	inside the same face $f$ of $\mathcal{E}$. This implies that there
	exists in $G_{\Psi(e_i)}$ a crossing between a path
	$p'=(a,\dots,v_{\Psi(e_i)}(e_i),\dots,b)$ and
	$p''=(c,\dots,v_{\Psi(e_i)}(e_j),\dots,d)$ only containing exclusive
	edges of $G_{\Psi(e_i)}$ such that $a,c,b$, and $d$ apper in this
	order in the face of $\mathcal{E}_{\Psi(e_i)}$ corresponding to
	$f$. Thus, both assumptions contradict the fact that \instance{}
	admits an \SDB.

	%%%%%%%%%%%%%%%%%%%%%%%%%%%%%%%%%%%%%%%%%%%%%%%%%%%%%%%%%%%%%%%%%%%%%%
	
	% Other direction
	Suppose that \instance{} admits an \SDB, that is, there exist a
	planar embedding $\mathcal{E}$ of $\beta$ and an assignment function
	$A_\mathcal{E}: E \cup \mathcal{Q} \rightarrow F(\mathcal{E})$ such
	that, for any two paths $p'=(a,\dots,b)$ and $p''=(c,\dots,d)$ with
	$\{a,b,c,d\} \in \beta$ and $\Psi(e_j)-\Psi(e_i)< \omega$, for every
	edge $e_i \in p'$ and $e_j \in p''$ with $i<j$, it holds that
	$A_\mathcal{E}(e_i)\neq A_\mathcal{E}(e_j)$.
	% CONSTRUCTION
	We show how to construct a SEFE \sefeksolution{m+1} of
	\sefekinstance{m+1}.
	
	For $i=1,\dots,m+1$, we set the embedding $\mathcal{E}_i$ of $\beta$
	to $\mathcal{E}$. For $i=1,\dots,m+1$ and for each edge $e \in E$,
	% such that $0 \leq i-\Psi(e)<\omega$,
	we assign each vertex $v_k(e) \in D(e)$ to the face of
	$\mathcal{E}_i$ that corresponds to the face of $\mathcal{E}$ edge
	$e$ is assigned to, that is,
	$A_{\mathcal{E}_i}(v_k(e))=A_{\mathcal{E}}(e)$. Also, for each edge
	$e=(u,v) \in E$, we assign edges $(u,v_k(e))$ and $(v_k(e),v)$ to
	face $A_{\mathcal{E}_k}(v_k(e))$, with $\Psi(e)\leq k <
	min(\Psi(e)+\omega,m+1) $.  Further, for each edge $e=(u,v) \in E$,
	we assign edges $(v_{\Psi(e)},v_k(e))$ to face
	$A_{\mathcal{E}_{m+1}}(v_k(e))$, with $\Psi(e)< k <
	min(\Psi(e)+\omega,m+1)$.  Finally, for $i=1,\dots,m+1$ and for each
	vertex $v \in \mathcal{Q}$, we set
	$A_{\mathcal{E}_i}(v)=A_{\mathcal{E}}(v)$.
	
	In order to prove that \sefeksolution{m+1} is a SEFE of
	\sefekinstance{m+1} we show that (i) $\mathcal{E}_i$ is a planar
	embedding of $\beta \in G_i$ (ii) all embeddings $\mathcal{E}_i$
	coincide, (ii) there exists no crossing in $\Gamma_i$ involving the
	exclusive edges of any graph $G_i$, and (iv) each isolated vertex
	$v$ of $G_\cap$ is such that
	$A_{\mathcal{E}_i}(v)=A_{\mathcal{E}_j}(v)$, with $i \neq j$.  Since
	$\mathcal{E}$ is planar by hypothesis and since
	$\mathcal{E}_i=\mathcal{E}$, condition (i) is trivially
	verified. Further, by construction, conditions (ii) and (iv), are
	also satisfied.
	Assume that condition (iii) does not hold. In this case, either an
	exclusive edge $(v_i(e),w)$ of $G_i$ crosses an edge of $\beta$ or
	there exists a crossing between two exclusive edges $(v_i(e_1),p)$
	and $(v_i(e_2),q)$ of $G_i$ inside the same face of $\mathcal{E}_i$.
	In the former case, there must exists in $G_i$ a path
	$p_0=(a,v_{i}(e),b)$ composed of exclusive edges of $G_i$ connecting
	two vertices ${a,b} \in \beta$ (not necessarily different from $w$)
	that lie on the boundary of different faces of
	$\mathcal{E}_i$. However, this would imply that $G_\cup$ contains a
	path $p^*_0=(a,\dots,b)$ containing edge $e$ and only consisting of
	edges $e_k$ with $0 \leq i - \Psi(e_k) < \omega$, whose endpoints
	$a$ and $b$ lie on different faces of $\mathcal{E}$.
	In the latter case, there must exist two vertex-disjoint paths
	$p_1=(a,\dots,v_i(e_1),\dots,b)$ and
	$p_2=(c,\dots,v_i(e_2),\dots,d)$ of exclusive edges of $G_i$
	contained in a face $f$ of $\mathcal{E}_i$ connecting vertices
	${a,b} \in f$ and ${c,d} \in f$, respectively, such that $a,c,b$,
	and $d$ appear in this order along $f$.  However, this would imply
	that $G_\cup$ contains two paths $p^*_1=(a,\dots,b)$ and
	$p^*_2=(c,\dots,d)$ with endpoints in $\beta$ containing edges $e_1$
	and $e_2$, respectively, and only containing edges $e_k$ in $E$ with
	$0 \leq i - \Psi(e_k) < \omega$ that lie inside the face $f^*$ of
	$\mathcal{E}$ corresponding to face $f$ of $\mathcal{E}_i$ and whose
	endpoints alternate along the boundary of $f^*$. Thus, both
	assumptions contradict the fact that \instance{} admits an \SDB.
	
	It is easy to see that instance \sefekinstance{m+1} can be
	constructed in time $O(n + \omega{}m)$. In fact, the
	construction of the common graph $G_\cap$ takes $O(n)$-time, since
	the backbone graph $G$ is planar. Also, each graph $G_i$ can be
	encoded as the union of a pointer to the encoding of $G_\cap$ and of
	the encoding of its exclusive edges. Further, each graph $G_i$, with
	$i = 1,\dots, m $, has at most $\omega$ exclusive edges, and graph
	$G_{m+1}$ has at most $\omega{}m$ exclusive edges. This
	concludes the proof of the lemma.
\end{proof}

Lemma~\ref{le:starreduction} provides a straight-forward technique to
decide whether a star istance \instance{} of \spBp admits a \SDB. First,
transform instance \instance{} into an equivalent instance
\sefekinstance{m+1} of SEFE of $m+1$ graphs with sunflower
intersection and such that the common graph consists of disjoint
$2$-connected components, by applying the reduction described in the
proof of Lemma~\ref{le:starreduction}. Then, apply to instance
\sefekinstance{m+1} the algorithm by Bl{\"a}sius {\em et
  al.}~\cite{bkr-seeorpc-13} that tests instances of SEFE with the
above properties in linear time. Thus, we obtain the following
theorem.

\begin{theorem}\label{th:algo-star}
  Let \instance{} be an star instance of \spBp. There exists an
  $O(n+ \omega{}m)$-time algorithm to decide whether \instance{} admits an \SDB.
\end{theorem}

\subsection{Unit window size}\label{se:unit}
 In this section we describe a
polynomial-time algorithm to test whether an instance \instance{} of
\spBp admits an \SDB for $\omega=1$.
Observe that, in the case in which $\omega=1$, the \spBp Problem
equals to the problem of deciding whether an embedding of
the backbone graph exists such that the endpoints of each
edge of the stream lie on the boundary of the same face of such an embedding.

Let $\mathcal{G}_1,\dots,\mathcal{G}_{1(G)}$ be the connected components of the backbone graph $G$.
Given an embedding $\mathcal{E}$ of $G$, we define
the set $F(\mathcal{E})$  of facial cycles of $\mathcal{E}$ as
the union of the facial cycles of the embeddings $\mathcal{E}_i=\mathcal{E}|_{\mathcal{G}_i}$ of
each connected component $\mathcal{G}_i$ of $G$ in
$\mathcal{E}$.
We first prove an auxiliary lemma which allows us to focus our attention only on instances whose backbone graph contains at most one non-trivial connected component.

\newcommand{\lemmaCONNECTED}{  
Let \instance{} be an instance of \spBp. There exists a set of
  instances \instance{i} whose backbone graph $G(V_i,S_i)$ contains at
  most one non-trivial connected component $\mathcal{G}_i$ %, with $i=1,\dots,1(G)$,
  such that \instance{} admits a \SDB with $\omega=1$ if and only if
  all instances \instance{i} admit a \SDB with $\omega=1$. Further,
  such instances can be constructed in $O(n+m)$ time.}
\begin{lemma}\label{le:connected}
\lemmaCONNECTED
\end{lemma}

   \begin{proof}
%   	Contracting an edge $(u,v)$ in a graph $H$ denotes the operation
%   	of first removing from $H$ vertices $u$ and $v$ and their incident
%   	edges, then introducing a new vertex $w$ that is incident to all
%   	the edges $u$ and $v$ used to be incident to, except edge $(u,v)$,
%   	and finally removing multiple edges, if any.
%   	
   	We construct instances \instance{i} starting from $G_\cup$ in two
   	steps. To ease the description, we assume that each vertex $v \in
   	V$ is initially associated with an index $l(v)$ corresponding to the connected component of $G$ vertex $v$ belongs to, that is,
   	$l(v)=i$ if $v \in V(\mathcal{G}_i)$.
   	First, we recursively contract each edge
   	$(u,v)$ of  $G_\cup$ with $\{u,v\} \subseteq V(\mathcal{G}_i)$ to a single vertex $w$ and set $l(w)=i$, for
   	$i=1,\dots,1(G)$. Thus, obtaining an auxiliary graph $H$ on $1(G)$
   	vertices. Then, we obtain instances \instance{i} from $H$ by
   	recursively uncontracting each vertex $w$ with $l(w)=i$, for
   	$i=1,\dots,1(G)$. Note that, by construction, $\mathcal{G}_i \subseteq 	G(V_i,S_i)$.
   	
   	Observe that, the construction of $H$ requires
   	$O(n+m)$ time. Further, the construction of each instance
   	\instance{i} can be performed in $O(n_i+m_i)$ time, where
   	$n_i=|V(\mathcal{G}_i)|$ and $m_i$ is the number of edges in $E$
   	that are incident to a vertex of $G_i$, which sums up to
   	$O(n+m)$ time in total for all $1 \leq i \leq 1(G)$. Thus, proving the $O(n+m)$
   	running time of the construction.
   	
   	The necessity is trivial. In order to prove the sufficiency, assume that
   	all instances \instance{i} admit a \SDB for $\omega=1$.  
   	%
%   	\red{
   	Intuitively, a \SDBp{1} $\Gamma$ of the original
   	instance can be obtained, starting from a \SDBp{1} $\Gamma_i$ of any \instance{i}, by
   	recursively replacing the drawing of each isolated vertex $v_j \in
   	\mathcal{Q}_i$ with the \SDBp{1} $\Gamma_j$ of \instance{j} (after, possibly, promoting a different face to be the outer face of $\Gamma_j$) .
   	For a complete example, see Fig.~\ref{fig:connected-composition}.
   	The fact that $\Gamma$ is a \SDBp{1} of \instance{} derives from the fact that
   	each $\Gamma_i$ is a \SDBp{1} of \instance{i}, that in a \SDBp{1} crossings among edges in $E$ do not matter,  and that, by the
   	connectivity of the union graph, the assignment of the isolated
   	vertices in $\mathcal{Q}_i$ to the faces of the embedding
   	$\mathcal{E}_i$ of $\mathcal{G}_i$ in $\Gamma_i$ must be such that any two isolated vertices connected by a path of
   	edges of the stream $E_i$ lie inside the same face of
   	$\mathcal{E}_i$. In the following, we prove this direction more formally.
%   	}
   	
   	We denote by
   	$(\mathcal{E}_i, C_{\mathcal{E}_i})$ the solution of instance
   	\instance{i}, where $\mathcal{E}_i$ is a planar embedding of
   	$\mathcal{G}_i$ and $C_{\mathcal{E}_i}: F(\mathcal{E}_i)
   	\rightarrow 2^{\mathcal{Q}_i}$ is an assignment of the set of
   	isolated vertices $\mathcal{Q}_i$ of $G(V_i,S_i)$ to the set of faces of
   	$\mathcal{E}_i$, denoted by $F({\mathcal{E}_i})$.
   	We now show how to extend the solutions $(\mathcal{E}_i,
   	C_{\mathcal{E}_i})$ of instances \instance{i}, with
   	$i=1,\dots,1(G)$, to a solution $\langle\mathcal{E}, C_\mathcal{E}\rangle$ of
   	instance \instance{}, where $\mathcal{E}$ is a planar
   	embedding of $G$ defining the set of facial cycles and
   	$C_{\mathcal{E}}: F(\mathcal{E}) \rightarrow 2^{\{1,\dots,1(G)\}}$ is
   	an assignment of the connected components of $G$ to the faces of
   	$\mathcal{E}$. 
   	
   	To obtain $\mathcal{E}$, we set the rotation scheme of each vertex
   	$v$ of $G$ in $\mathcal{E}$ to the rotation scheme of $v$ in the
   	embedding $\mathcal{E}_i$ of the component $\mathcal{G}_i$ of the backbone graph
   	$G$ containing $v$.
   	Clearly, the set of facial cycles $F(\mathcal{E})$ of
   	$\mathcal{E}$ is equal to the union of the set of facial cycles of
   	each $\mathcal{E}_i$, that is, for each face $f \in \mathcal{E}$,
   	we have that $f$ belongs to $\mathcal{E}_i$ for some $1\leq i\leq
   	1(\mathcal{G})$.
   	
   	The assignment function $C_\mathcal{E}$ can be defined as follows.
   	Initialize $C_\mathcal{E}(f)=\emptyset$, for each facial cycle $f$
   	in $F(\mathcal{E})$.  Then, consider each pair of connected components $\mathcal{G}_i$ and $\mathcal{G}_j$ of the backbone graph and, for each facial cycle $f$ in
   	$F(\mathcal{E}) \cap  F(\mathcal{E}_j)$,
   	set $C_\mathcal{E}(f)=C_\mathcal{E}(f)\cup i$ if
   	$i \in C_{\mathcal{E}_j}(f)$.
   	
   	We now prove that $(\mathcal{E},C_{\mathcal{E}})$ is a solution
   	for \instance{}. Since each $\mathcal{E}_i$ is a planar embedding,
   	then $\mathcal{E}$ is also planar. We just need to prove that for
   	every two faces $f'$ and $f''$ of $\mathcal{E}$ either (i)
   	$C_\mathcal{E}(f')\subseteq C_\mathcal{E}(f'')$, or (ii)
   	$C_\mathcal{E}(f'')\subseteq C_\mathcal{E}(f')$, or (iii)
   	$C_\mathcal{E}(f')\cap C_\mathcal{E}(f'')= \emptyset$. Clearly, if
   	$f',f'' \in \mathcal{E}_i$ for some $i$, exactly one of (i), (ii),
   	and (iii) must hold, as otherwise $(\mathcal{E}_i,
   	C_{\mathcal{E}_i})$ would not be a solution of \instance{i}.  We
   	prove that there exist no $f' \in \mathcal{E}_i$ and $f'' \in
   	\mathcal{E}_j$ with $i\neq j$ such that neither (i), (ii), or
   	(iii) holds.  We distinguish three cases according to whether $j
   	\in C_{\mathcal{E}_i}(f')$, or $i \in C_{\mathcal{E}_j}(f'')$, or
   	$j \notin C_{\mathcal{E}_i}(f') \wedge i \notin
   	C_{\mathcal{E}_j}(f'')$. By the connectivity of the union graphs
   	of each instance and by the fact that $(\mathcal{E}_i,
   	C_{\mathcal{E}_i})$ and $(\mathcal{E}_j, C_{\mathcal{E}_j})$ are
   	\SDB of \instance{i} and \instance{j}, respectively, we have that:
   	(i) must hold, if $i \in C_{\mathcal{E}_j}(f'')$; (ii) must hold,
   	if $j \in C_{\mathcal{E}_i}(f')$; and (iii) must hold, if $j \notin
   	C_{\mathcal{E}_i}(f') \wedge i \notin C_{\mathcal{E}_j}(f'')$. This concludes the proof of the lemma.
   \end{proof}

By Lemma~\ref{le:connected}, in the following we only
consider the case in which the backbone graph consists of a single
non-trivial connected component plus,
possibly, isolated vertices.  
We now present a simple recursive algorithm to test instances with this property.

\noindent{\\\underline{Algorithm \algorithmbf.}}
\begin{itemize}
\item[$\circ$ {\em INPUT:}] an instace $I=$~\instance{} of the {\sc \spBp} Problem
with $\omega=1$ with union graph $G_\cup$ such that $G$ contains at most one non-trivial connected component. 
\item[$\circ$ {\em OUTPUT:}] \texttt{YES}, if \instance{} is positive, or \texttt{NO}, otherwise.
\end{itemize}

{\bf BASE CASE~1:} instance $I$ is such that $2(G)=0$, that is, every connected component of $G$ is an isolated vertex.
Return \texttt{YES}, as instances of this kind are trivially positive. 

{\bf BASE CASE~2:} instance $I$ is such that (i) $2(G)=1$,
that is, the backbone graph $G$ consists of a single
$2$-connected component plus, possibly, isolated vertices and (ii) no
edge of the stream connects any two isolated vertices.
In this case, apply the algorithm of Theorem~\ref{th:algo-star} to decide $I$ and return \texttt{YES}, if the test succeeds, or \texttt{NO}, otherwise.

{\bf RECURSIVE STEP:} instance $I$ is such that either (\texttt{CASE~R1})
$2(G)=1$ and there exists edges of the stream between pairs of
isolated vertices or (\texttt{CASE~R2}) \mbox{$2(G)>1$}. First, replace instance $I$ with two smaller instances
$I^\diamond=$~\instance{\diamond} and $I^\circ=$~\instance{\circ}, as described below.
Then, %apply algorithm \algorithm to $I^\diamond$ and $I^\circ$ and 
return \texttt{YES}, if $\algorithm(I^\diamond)=$ 
\noindent
$\algorithm(I^\circ)$ $=$ \texttt{YES}, or \texttt{NO}, otherwise.

\begin{itemize}
\item[{ \em CASE~R1}.]
  Instance $I^\diamond$ is obtained from $I$ by recursively contracting every
  edge $(u,v)$ of $G_\cup$ with $\{u,v\} \nsubseteq V(\mathcal{G})$.
  Instance $I^\circ$ is obtained from $I$ by recursively contracting every
  edge $(u,v)$ of  $G_\cup$ with $\{u,v\} \subseteq V(\mathcal{G})$.
\item[{ \em CASE~R2}.]
Let $\mathcal{G}$ be the unique non-trivial connected component of $G$, let $T$ be the block-cutvertex tree of $\mathcal{G}$ rooted at any block, and let $\beta$ be any leaf block in $T$. Also, let $v$ be
  the parent cutvertex of $\beta$ in $T$.
  We first construct an auxiliary equivalent instance
  $I^*=$~\instance{*} starting from $I$ and then obtain instances $I^\diamond$
  and $I^\circ$ from $I^*$, as follows. See Fig.~\ref{fig:leafBlockSplit}
  for an illustration of the construction of instance $I^*$.
  \begin{figure}[tb]
    \centering
    \begin{subfigure}{0.48\textwidth}
      \centering
      \includegraphics[height=0.5\textwidth]{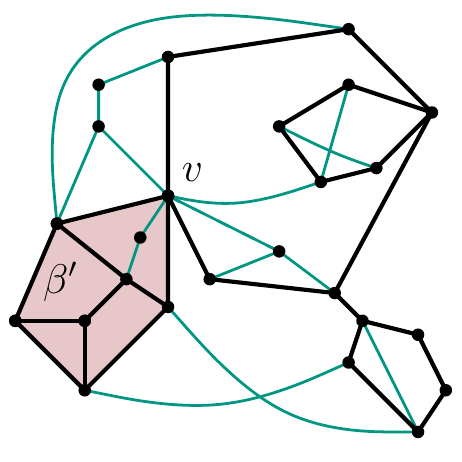}
    \end{subfigure}
    \begin{subfigure}{0.48\textwidth}
      \centering
      \includegraphics[height=0.5\textwidth]{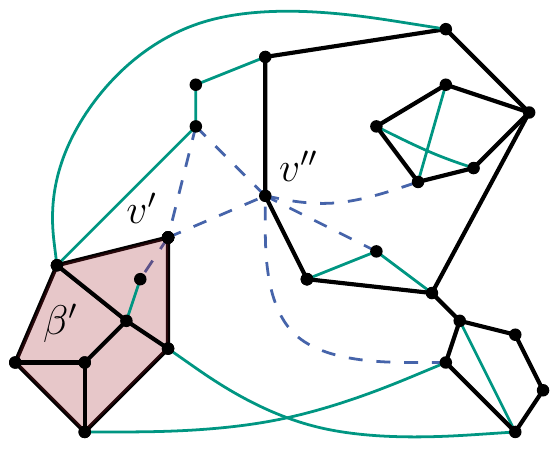}
    \end{subfigure}
    \caption{(a) Instance $I$ and (b) instance $I^*$ obtained in {\em CASE~R2}
      of Algorithm \algorithm. Edges of the backbone
      graph are black thick curves; edge of the stream are
      green thin curves; and edges of the stream incident to $v'$ and
      $v''$ in $I^*$ are blue dashed
      curves.}\label{fig:leafBlockSplit}
  \end{figure}
  Initialize $I^*$ to $I$.
  Replace vertex $v$ in $V_*$ with two vertices $v'$
  and $v''$ and make (i) $v'$ adjacent to all the vertices of
  $\beta$ vertex $v$ used to be adjacent to and (ii) $v''$ adjacent
  to all the vertices in $V(\mathcal{G})\setminus V(\beta)$ vertex $v$
  used to be adjacent to.
  Then, replace each edge $(v,x)$ of $E^*$ with edge $(v',x)$, if $x
  \in V(\beta)$ or if $x \in \mathcal{Q}^*$ and there
  exists a path composed of edges of the stream connecting $x$ to a
  vertex $y\neq v \in V(\beta)$, and edge $(v'',x)$, if $x \in V(\mathcal{G})
  \setminus V(\beta)$ or if $x \in \mathcal{Q}^*$ and there exists a
  path composed of edges of the stream connecting $x$ to a vertex
  $y \neq v \in V(\mathcal{G}) \setminus V(\beta)$. 
%Recall that, $\mathcal{Q}^*$ denotes the set of isolated vertices of $G(V_*,E_*)$.
  %
  Finally, add edge $(v',v'')$ to $E^*$.  
  It is easy to see that
  instances $I$ and $I^*$ are equivalent.

  Instance $I^\diamond$ is obtained from $I^*$ by recursively contracting
  every edge $(u,v)$ of  $G^*_\cup$ with $u,v \nsubseteq V(\beta)$, where $G^*_\cup$ is the union graph of $I^*$.
  Instance $I^\circ$ is obtained from $I^*$ by recursively contracting
  every edge $(u,v)$ of  $G^*_\cup$ with $\{u,v\} \subseteq V(\beta)$.

  % Instance $I^\circ$ is obtained from $G^*_\cup$ by contracting
  % $\upbeta'$ to a single vertex $v$ and by making $v$ adjacent to
  % the same vertices any vertex $v'\in \upbeta'$ used to be adjacent
  % to in $G^*_\cup$.
\end{itemize}

\begin{theorem}
  Let \instance{} be an instance of \spBp. There
  exists an $O(n+m)$-time algorithm to decide whether
  \instance{} admits an \SDB for $\omega=1$.
\end{theorem}

\begin{proof}
  The algorithm runs in two steps, as follows.
  \begin{itemize}
  \item {\bf STEP~1} applies the reduction illustrated in the proof of
    Lemma~\ref{le:connected} to \instance{} to construct
    $1(G)$ instances \instance{i} such that the backbone
    graphs $G(V_i,S_i)$  contain at most one non-trivial connected component.
  \item {\bf STEP~2} applies Algorithm \algorithm to every
    instance \instance{i} and return \texttt{YES}, if all such
    instances are positive, or \texttt{NO}, otherwise.
  \end{itemize}

  Observe that, the correctness of the presented algorithm follows
  from the correctness of Lemma~\ref{le:connected}, of Theorem~\ref{th:algo-star}, and of Algorithm \algorithm. 
  %The first two have already been proved. 
  We now prove the correctness for Algorithm \algorithm.
  Obviously, the fact that instances $I^\diamond$ and $I^\circ$ constructed in
  \texttt{CASE~R1} and \texttt{CASE~R2}  are both
  positive is a necessary and sufficient condition for instance $I$ to
  be positive. We prove termination by induction on the number
  $2(G)$ of blocks of the backbone graph $G$ of
  instance $I$, primarily, and on the number of edges of the stream
  connecting isolated vertices of the backbone graph, secondarily.
\begin{inparaenum}[(i)]
	\item
  If $2(G)=0$, then \texttt{BASE~CASE~1} applies and the
  algorithm stops;
\item
  if $2(G)=1$ and no two isolated vertices of the backbone
  graph are connected by an edge of the stream, then
  \texttt{BASE~CASE~2} applies and the algorithm stops;
  \item
  if $2(G)=1$ and there exist edges of the stream between
  any two isolated vertices of the backbone graph $G$, then,
  by \texttt{CASE~R1}, instance $I$ is split into (a) an instance $I^\diamond$ with
  $2(G(V_{\diamond},E_{\diamond}))=1$ and no edges of the stream connecting any two
  isolated vertices of the backbone graph $G(V_{\diamond},E_{\diamond})$, and (b) an
  instance $I^\circ$ with $2(G(V_{\circ},E_{\circ}))=0$;
\item
  finally, if $2(G)>1$, then, by \texttt{CASE~R2}, instance $I$
  is split into (a) an instance $I^\circ$ with $2(G(V_{\diamond},E_{\diamond}))=1$ and
  (b) an instance $I^\circ$ with $2(G(V_{\circ},E_{\circ}))=2(G)-1$.
\end{inparaenum}

  The running time easily derives from the fact that all instances
  \instance{i} can be constructed in $O(n+m)$-time and that the
  algorithm for star instances described in the proof of
  Theorem~\ref{th:algo-star} runs in $O(n+\omega{}m)$-time. This
  concludes the proof.
\end{proof}

\paragraph{Acknowledgments.} Giordano {Da Lozzo} was supported by the MIUR project AMANDA ``Algorithmics for MAssive and Networked DAta'', prot. 2012C4E3KT\_001. Ignaz \mbox{Rutter} was supported by a fellowship within the Postdoc-Program of the German Academic Exchange Service (DAAD). This work was done while the authors where visiting the Department of Applied Mathematics at Charles University in Prague.

   	\begin{figure}[tb]
   		\begin{subfigure}{.32\textwidth}
   			\centering
   			\includegraphics[page=1, width=\textwidth]{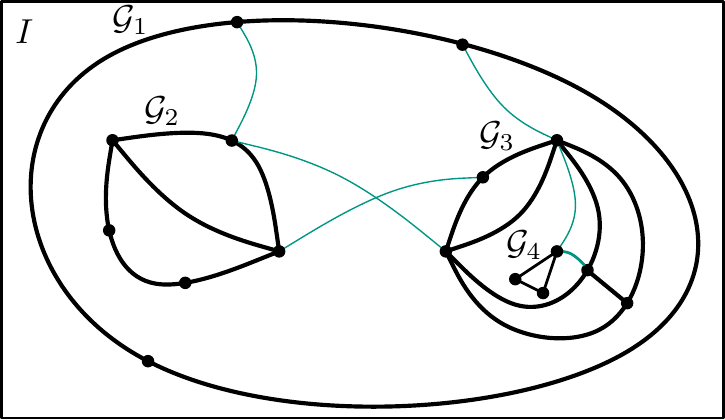}
   			\caption{}
   		\end{subfigure}
   		\begin{subfigure}{.32\textwidth}
   			\centering
   			\includegraphics[page=2, width=\textwidth]{img/connected-composition}
   			\caption{}
   		\end{subfigure}
   		\begin{subfigure}{.32\textwidth}
   			\centering
   			\includegraphics[page=3, width=\textwidth]{img/connected-composition}
   			\caption{}
   		\end{subfigure}
   		\\
   		\begin{subfigure}{.32\textwidth}
   			\centering
   			\includegraphics[page=4, width=\textwidth]{img/connected-composition}
   			\caption{}
   		\end{subfigure}
   		\begin{subfigure}{.32\textwidth}
   			\centering
   			\includegraphics[page=5, width=\textwidth]{img/connected-composition}
   			\caption{}
   		\end{subfigure}
   		\begin{subfigure}{.32\textwidth}
   			\centering
   			\includegraphics[page=6, width=\textwidth]{img/connected-composition}
   			\caption{}
   		\end{subfigure}
   		\\
   		\begin{subfigure}{.32\textwidth}
   			\centering
   			\includegraphics[page=7, width=\textwidth]{img/connected-composition}
   			\caption{}
   		\end{subfigure}
   		\begin{subfigure}{.32\textwidth}
   			\centering
   			\includegraphics[page=8, width=\textwidth]{img/connected-composition}
   			\caption{}
   		\end{subfigure}
   		\begin{subfigure}{.32\textwidth}
   			\centering
   			\includegraphics[page=9, width=\textwidth]{img/connected-composition}
   			\caption{}
   		\end{subfigure}
   		\\
   		\begin{subfigure}{.32\textwidth}
   			\centering
   			\includegraphics[page=10, width=\textwidth]{img/connected-composition}
   			\caption{}
   		\end{subfigure}
   		\begin{subfigure}{.32\textwidth}
   			\centering
   			\includegraphics[page=11, width=\textwidth]{img/connected-composition}
   			\caption{}
   		\end{subfigure}
   		\begin{subfigure}{.32\textwidth}
   			\centering
   			\includegraphics[page=12, width=\textwidth]{img/connected-composition}
   			\caption{}
   		\end{subfigure}
   		\caption{(a) Instance $I=$~\instance{} of \spBp with $\omega=1$, where $G$ consists of $4$ connected components $\mathcal{G}_1$, $\mathcal{G}_2$, $\mathcal{G}_3$, and $\mathcal{G}_4$. 
   			Edges of the backbone graph are black thick curves. Edges of the stream are green thin curves. 
   			Instances $I_1$ (b), $I_2$ (d), $I_3$ (f), and $I_4$ (h) obtained by applying the procedure described in the proof of Lemma~\ref{le:connected} to instance $I$. \SDBp{1} $\Gamma_1$ (c), $\Gamma_2$ (e), $I_3$ (g), and $I_4$ (i) of instances $I_1$, $I_2$, $I_3$, and $I_4$, respectively. (j) \SDBp{1} $\Gamma^{12}$ obtained by replacing the drawing of $\mathcal{G}_2$ in $\Gamma_1$ with $\Gamma_2$.  (k) \SDBp{1} $\Gamma^{123}$ obtained by replacing the drawing of $\mathcal{G}_3$ in $\Gamma^{12}$ with $\Gamma^{3}$.  (l) \SDBp{1} $\Gamma^{1234}$ obtained by replacing the drawing of $\mathcal{G}_4$ in $\Gamma^{123}$ with $\Gamma_4$.}\label{fig:connected-composition}
   	\end{figure}

\clearpage
{\bibliographystyle{splncs03} \bibliography{bibliography} }
\end{document}